%% file: main.tex
\documentclass{article}
\usepackage{fullpage}
\usepackage{graphicx} 
\usepackage{amsthm, amsmath, amssymb, amsfonts}
\usepackage{xcolor}
\usepackage{booktabs}
\usepackage{todonotes}
\usepackage{multirow}
\usepackage{hhline}
\usepackage{hyperref}

\usepackage{enumitem}
\setlist{topsep=3pt, itemsep=0pt}

\graphicspath{{figs/}}

\usepackage{algorithm}
\usepackage{algpseudocode}
\algtext*{EndFor}{}
\algtext*{EndIf}{}
\algtext*{EndWhile}{}

\newtheorem{theorem}{Theorem}
\newtheorem{lemma}[theorem]{Lemma}
\newtheorem{coro}[theorem]{Corollary}
\newtheorem{claim}[theorem]{Claim}

\newcommand{\Z}{\mathbb{Z}}
\newcommand{\R}{\mathbb{R}}

\newcommand{\calP}{\mathcal{P}}
\newcommand{\calS}{\mathcal{S}}
\newcommand{\calT}{\mathcal{T}}
\newcommand{\calX}{\mathcal{X}}

\newcommand{\ceil}[1]{\left\lceil#1\right\rceil}

\newcommand{\opt}{\mathrm{opt}}

\newcommand{\rmin}{\mathrm{in}}
\newcommand{\rmout}{\mathrm{out}}

\newcommand{\supp}{\mathrm{supp}}
\newcommand{\poly}{\mathrm{poly}}
\newcommand{\bfT}{\mathbf{T}}
\DeclareMathOperator{\union}{\bigcup}

\newcommand{\remove}[1]{}

\title{Complexity and Approximation Algorithms for Fixed Charge Transportation Problems}

\author{ 
Yong Chen \\ Department of Mathematics \\ Hangzhou Dianzi University, \\ Hangzhou, Zhejiang Province, China \\ 
\href{mailto:chenyong@hdu.edu.cn}{chenyong@hdu.edu.cn} 
\and Shi Li  \\ School of Computer Science, \\ Nanjing University, \\ Nanjing, Jiangsu Province, China \\ \href{mailto:shili@nju.edu.cn}{shili@nju.edu.cn} \and Zihao Liang \footnotemark[\value{footnote}] \\ School of Computer Science, \\ Nanjing University, \\ Nanjing, Jiangsu Province, China \\\href{zhliang@smail.nju.edu.cn}{zhliang@smail.nju.edu.cn}
}

\date{}
\begin{document}
    \maketitle
    \input{intro}

    \input{prelim}
    \input{overview}
    \input{PFCT-S-2-approx}

    \input{PFCT-U-1.2-approx}

    \input{hardness}

    \input{EPAS}
    \input{discussion}

    \appendix

    \input{appendix}

\end{document}

%% file: intro.tex
\begin{abstract}
    The Fixed Charge Transportation (FCT) problem models transportation scenarios where we need to send a commodity from $n$ sources to $m$ sinks, and the cost of sending a commodity from a source to a sink consists of a linear component and a fixed component. Despite extensive research on exponential time exact algorithms and heuristic algorithms for FCT and its variants, their approximability and computational complexity are not well understood.
    
    In this work, we initiate a systematic study of the approximability and complexity of these problems. When there are no linear costs, we call the problem the Pure Fixed Charge Transportation (PFCT) problem.  We also distinguish between cases with general, sink-independent, and uniform fixed costs; we use the suffixes ``-S'' and ``-U'' to denote the latter two cases, respectively.  This gives us six variants of the FCT problem.
    
    We give a complete characterization of the existence of $O(1)$-approximation algorithms for these variants.  In particular, we give $2$-approximation algorithms for FCT-U and PFCT-S, and a $(6/5 + \epsilon)$-approximation for PFCT-U. On the negative side, we prove that FCT and PFCT are NP-hard to approximate within a factor of $O(\log^{2-\epsilon} (\max\{n, m\}))$ for any constant $\epsilon > 0$, FCT-S is NP-hard to approximate within a factor of $c\log (\max\{n, m\})$ for some constant $c> 0$, and PFCT-U is APX-hard. 
    Additionally, we design an Efficient Parameterized Approximation Scheme (EPAS) for PFCT when parameterized by the number $n$ of sources, and an $O(1/\epsilon)$-bicriteria approximation for the FCT problem, when we are allowed to violate the demand constraints for sinks by a factor of $1\pm \epsilon$.
\end{abstract}

\thispagestyle{empty}\newpage\setcounter{page}{1}

\section{Introduction}
\label{sec:intro}
In 1954, Hirsch and Dantzig  \cite{HD54, HD68} introduced the Fixed Charge Transportation (FCT) problem, where the cost of transportation from a source to a sink contains a linear component, which is proportional to the amount transported, and a fixed cost, which is incurred whenever the transportation occurs. 
The objective is to allocate all supplies of the commodity from sources to sinks in a way that minimizes the total cost, which is the sum of linear and fixed costs.

The FCT problem is a special case of the single-commodity uncapacitated fixed-charge network problem \cite{RW, OW}, which itself belongs to a broader class of network design problems \cite{A, FG, RKOW}. The single-commodity uncapacitated fixed-charge network problem is defined on an arbitrary directed graph \( (V, E) \), where each edge \( e \in E \) has a linear cost and a fixed cost as in the FCT problem. The commodity can be transported only along edges in \( E \). 
The goal is to match the given supplies and demands so as to minimize the sum of variable and fixed costs. The FCT problem can be seen as a special case of this problem on complete bipartite directed graphs from sources to sinks. As a byproduct in this paper, we show that the two problems are indeed equivalent.


When there are only fixed costs (i.e., linear costs are zero), the FCT problem is called the Pure Fixed Charge Transportation (PFCT) problem \cite{FM, GL}. The PFCT problem arises in scenarios where transportation costs are negligible compared to the costs of establishing transportation routes, such as pipelines and electric power systems. Fisk and McKeown \cite{FM} proposed a search algorithm to solve the PFCT problem exactly, and Göthe-Lundgren and Larsson \cite{GL} reformulated it as a set-cover problem, and designed heuristic algorithms by relating it to the maximum flow problem. 

Despite considerable research on exact algorithms \cite{M,G,AA,ZLRP,MR}, branch-and-bound algorithms \cite{KU,BGK,CE,RBM}, and heuristic algorithms \cite{WH,SAMD,A09,BRT} for the FCT and PFCT problems, their computational complexity and approximation algorithms were not well studied. To the best of our knowledge, the only known results in this direction are that both FCT \cite{GP} and PFCT \cite{K05,SFC} are NP-hard.

In this paper, we initiate a systematic study of the approximability and complexity of FCT problems. Before we state our results, we give formal definitions of the problems. 
    
\subsection{Formal Definitions of Problems}

\label{subsec:formal-definition}

In the Fixed Charge Transportation (FCT) problem, we are given a set \( S \) of \( n \) sources and a set \( T \) of \( m \) sinks. Each source \( i \in S \) has a supply of \( a_i \in \mathbb{Z}_{>0} \) units for a commodity, and each sink \( j \in T \) has a demand of \( b_j \in \mathbb{Z}_{>0} \) units. It is guaranteed that \( \sum_{i \in S} a_i = \sum_{j \in T} b_j \). For every pair \( i \in S \) and \( j \in T \), there is an edge \( ij \) with a fixed cost \( f_{ij} \in \mathbb{R}_{\geq 0} \) and a linear cost ratio \( c_{ij} \in \mathbb{R}_{\geq 0} \). 
The objective is to satisfy all demands by sending the commodity from sources to sinks while minimizing the total cost. The formal mathematical formulation is as follows:
\[\textstyle
\min \quad \sum_{i \in S, j \in T} \left( 1_{x_{ij} > 0} \cdot f_{ij} + c_{ij} \cdot x_{ij} \right) \qquad \text{s.t.} \qquad x \in \calX,
\]
where throughout the paper we use 
$$ \textstyle
\calX := \left\{x \in \R_{\geq 0}^{S \times T}\quad :\quad
\sum_{j \in T} x_{ij} = a_i, \forall i \in S; \quad 
\sum_{i \in S} x_{ij} = b_j, \forall j \in T 
\right\}
$$
to denote the set of all feasible solutions to the given FCT instance. \medskip

The Pure Fixed Charge Transportation (PFCT) problem is the special case where there are no linear costs, i.e., \( c_{ij} = 0 \) for every \( i \in S \) and \( j \in T \). We shall show that FCT and PFCT problems capture the Directed Steiner Tree (DST) and Set Cover problems respectively, and thus do not admit constant approximation unless P $=$ NP. To gain a more refined understanding of its complexity, we study two special cases of the fixed cost vector \( f \):  
\begin{itemize}
    \item \emph{Sink-independent fixed costs}. The fixed costs \( (f_{ij})_{i \in S, j \in T} \) are said to be sink-independent if there exists a cost vector \( f \in \mathbb{R}_{\geq 0}^S \) such that \( f_{ij} = f_i \) for all \( i \in S \) and \( j \in T \). We denote this special case using the suffix ``-S''.  
    \item \emph{Uniform fixed costs}. The fixed costs are uniform if \( f_{ij} = 1 \) for all \( i \in S \) and \( j \in T \). This case is denoted by the suffix ``-U''.
\end{itemize}

So, combining these variations—whether linear costs are present and the properties of the fixed costs—results in six variants of the FCT problem, denoted as FCT,  FCT-S, FCT-U, PFCT,  PFCT-S and PFCT-U.

\subsection{Our Results: $O(1)$-Approximability}

We provide a complete characterization of the existence of $O(1)$-approximation algorithms for the six variants of FCT. The results are listed in Table~\ref{tab:results} for convenience.
\def\arraystretch{1.5}
\begin{table}[h]
    \centering
    \begin{tabular}{||c||c|c||c|c||} \hhline{#=#==#==#}
       \multirow{2}{*}{}  & \multicolumn{2}{c||}{FCT} & \multicolumn{2}{c||}{PFCT}\\\cline{2-5}
         & Approx.~Ratio & Hardness & Approx.~Ratio & Hardness\\\hhline{#=#=|=#=|=#}
        general &   & $\Omega(\log^{2-\epsilon} (mn))$ &  & $\Omega(\log^{2-\epsilon} (mn))$ \\\hline
        -S &  & $\Omega(\log (mn))$ & $2$ & APX-hard \\\hline
        -U & $2$ & APX-hard & $\frac65 + \epsilon$ & APX-hard \\\hhline{#=#=|=#=|=#}
    \end{tabular}
    \caption{Approximability and Hardness of Variants of FCT.}
    \label{tab:results}
\end{table}

We first discuss the approximation results. For both PFCT-S and FCT-U, we give $2$-approximation algorithms:
\begin{theorem}
    \label{thm:PFCT-S-2}
    There is a $2$-approximation algorithm for the Pure Fixed Charge Transportation with Sink-Independent Fixed Costs (PFCT-S) problem. 
\end{theorem}

\begin{theorem}
    \label{thm:FCT-U-2}
    There is a $2$-approximation algorithm for the Fixed Charge Transportation with Uniform Fixed Costs (FCT-U) problem. 
\end{theorem}

For the most restricted problem PFCT-U, we present a better approximation algorithm:
\begin{theorem}
    \label{thm:PFCT-U-65}
    For any constant $\epsilon > 0$, there is a $(\frac65 + \epsilon)$ approximation algorithm for the Pure Fixed Charge Transportation with Uniform Fixed Costs (PFCT-U) problem.
\end{theorem}

Then we complement the results by showing the other variants are NP-hard to approximate within a constant factor.  First, we give a reduction from the Directed Steiner Tree (DST) problem to the PFCT problem. As the polynomial-time approximability of DST remains open,  we include the reduction in the theorem statement; any future improvements in the hardness of DST automatically carry over to PFCT. 
\begin{theorem}
    \label{thm:PFCT-DST-hard}
    Let $\alpha: \Z_{> 0} \to \R_{> 0}$ be a monotone function.  If there is a polynomial-time $\alpha(\max\{n, m\})$-approximation algorithm for the Pure Fixed Charge Transportation (PFCT) problem in polynomial time, then there is a polynomial-time $\alpha(n)$-approximation for the Directed Steiner Tree (DST) problem. 
\end{theorem}

Using the $\Omega(\log^{2 - \epsilon} (n))$-hardness of DST \cite{HK03}, this gives us
\begin{coro}
    \label{coro:PFCT-DST-hard}
    It is NP-hard to approximate the Pure Fixed Charge Transportation (PFCT) problem within a factor of $O(\log^{2 - \epsilon} (\max\{n, m\}))$, for any constant $\epsilon > 0$.
\end{coro}
Since PFCT is a special case of FCT, the results extend to FCT as well.  \medskip

We provide a reduction from the Set Cover problem to FCT-S, proving the following result:
\begin{theorem}
    \label{thm:FCT-S-Set-Cover-hard}
    There is a constant $c > 0$ such that the following holds. It is NP-hard to approximate the Fixed Charge Transportation with Sink-Independent Fixed Costs (FCT-S) problem within a factor of $c\ln (\max\{n, m\})$. 
\end{theorem}

Finally, we rule out the possibility of a PTAS for the PFCT-U problem, the most restricted variant we study. This is achieved via a reduction from the 3-Dimensional Matching with Bounded Frequency (3DM-B) problem.
\begin{theorem}
    \label{thm:PFCT-U-APX-hard}
    There is an absolute constant $c > 1$ such that Pure Fixed Charge Transportation with Uniform Fixed Costs (PFCT-U) problem does not have a $c$-approximation, unless NP $=$ BPP. 
\end{theorem}
The APX-hardness also extends to the more general problems FCT-U and PFCT-S. \medskip

\subsection{Our Results: Efficient Polynomial Time Approximation Scheme (EPAS) and Bicriteria Approximation Algorithm}

We then consider approximation algorithms parameterized by the number $n$ of sources. We show that there is an Efficient Parameterized Approximation Scheme (EPAS) for the PFCT problem: 
\begin{theorem}
    \label{thm:PFCT-EPAS}
    Given any constant $\epsilon > 0$, there is a $g(n, \epsilon) \cdot \poly(m)$-time $(1+\epsilon)$-approximation algorithm for the Pure Fixed Charge Transportation (PFCT) problem, for some efficiently computable function $g$. 
\end{theorem}

For the FCT problem, we show a simple $(O(1/\epsilon), 1\pm\epsilon)$-bicriteria approximation algorithm in polynomial time. Such algorithms output solutions which may violate the sink demand constraints by a factor of $1\pm \epsilon$. 
\begin{theorem}
    \label{thm:FCT-bicriteria}
    Given a Fixed Charge Transportation (FCT) instance defined by $S, T, (a_i)_{i}, (b_j)_{j}, (f_{ij})_{i, j}, (c_{ij})_{i, j}$, and a constant $\epsilon > 0$, we can efficiently output a vector $x \in \R_{\geq 0}^{S \times T}$ subject to $\sum_{j \in T}x_{ij} = a_i$ for every $i \in S$, $\sum_{i \in S} x_{ij} \in (1\pm \epsilon)b_j$ for every $j \in T$, 
    and $\sum_{i \in S, j \in T}\big(1_{x_{ij} > 0} \cdot f_{ij} + c_{ij} x_{ij}\big)$ is at most $O(1/\epsilon)$ times the optimal cost of the instance. 
\end{theorem}

\noindent{\bf Organization}\ \ The rest of the paper is organized as follows. In Section~\ref{sec:prelim}, we introduce the problems used in our hardness reductions. We provide an overview of our techniques in Section~\ref{sec:overview}. In Section~\ref{sec:PFCT-S-2-approx}, we describe the $2$-approximation for PFCT-S, proving Theorem~\ref{thm:PFCT-S-2}. 
In Section~\ref{sec:PFCT-U-65-approx}, we prove Theorem~\ref{thm:PFCT-U-65} by describing the $(\frac65 + \epsilon)$-approximation for PFCT-U. We show the hardness of the PFCT and PFCT-U problems in Sections~\ref{sec:PFCT-hardness} and \ref{sec:PFCT-U-from-3DM-B}, which prove Theorems~\ref{thm:PFCT-DST-hard} and \ref{thm:PFCT-U-APX-hard} respectively. 
The Efficient Parameterized Approximation Scheme (EPAS) for PFCT, which prove Theorems~\ref{thm:PFCT-EPAS} is given in Sections~\ref{sec:EPAS}.  Due to the page limit, we give the simple $2$-approximation for FCT-U (proof of  Theorem~\ref{thm:FCT-U-2}), the Set Cover hardness of FCT-S (proof of Theorem~\ref{thm:FCT-S-Set-Cover-hard}), and the bi-criteria approximation for FCT (proof of Theorem~\ref{thm:FCT-bicriteria}) in the appendix.

%% file: prelim.tex
\section{Preliminaries}

\label{sec:prelim}


We define several problems that we use in our reductions for proving the hardness results. \smallskip

\noindent{\bf Directed Steiner Tree}\ \ In the Directed Steiner Tree (DST) problem, we are given a directed graph $G=(V, E)$ of $n = |V|$ vertices, with a root $r \in V$ and a set $T \subseteq V \setminus \{r\}$ of $k$ terminals. Every edge $e \in E$ has a cost $c_e \in \R_{\geq 0}$. The goal of the problem is to find a subgraph $H$ of $G$ with the minimum cost, that contains a path from $r$ to $t$ for every $t \in T$. A minimal subgraph satisfying this property is a directed tree rooted at $r$ (i.e, a out arborescence), hence the name Directed Steiner Tree (DST). 

A simple reduction from Set Cover to DST gives a $(1-\epsilon)\ln k$-hardness for the latter problem, which translates to a hardness of $\Omega(\log n)$ in terms of $n$. Stronger inapproximability results are known: it is NP-hard to approximate DST within a factor of $O(\log^{2-\epsilon}n)$ for any constant $\epsilon > 0$ \cite{HK03}. Moreover, under the stronger assumption that  $\textrm{NP}\not\subseteq$ $\textrm{BPPTime}(2^{\log^{O(1)}n})$, the problem does not admit an $o(\log^2k/\log\log k)$-approximation algorithm. On the positive side, DST admits a quasi-polynomial time $O(\log^{2}n/\log\log n)$-approximation \cite{GLL22, GN22, CCC99}, and a $n^{O(1/\epsilon)}$-time $O(n^{\epsilon})$-approximation algorithm for any constant $\epsilon>0$ \cite{Z1997}. Whether there is a polynomial-time poly-logarithmic approximation for DST is a notorious open problem.\smallskip


\noindent{\bf Maximum 3-Dimensional Matching}\ \ 
In the 3-Dimensional Matching (3DM) problem, we are given a set $M \subseteq X \times Y \times Z$ of triples, where $X$, $Y$ and $Z$ are disjoint sets. The goal of the problem is to find the largest subset $M'\subseteq M$ such that no two tuples in $M'$ have the same coordinate. This is one of the Karp's 21 NP-hard problems.  

In our reduction, it is important to ensure that the number of times an element in $X \cup Y \cup Z$ appears in $M$ is bounded by an absolute constant $B$. We call such the problem the 3-Dimensional Matching with bounded frequency (3DM-B) problem.  The 3DM-B problem was proven to be APX-hard by Kann \cite{K91}, and later was shown to be NP-hard to approximate within a factor of $\frac{98}{97} - \epsilon$ for any constant $\epsilon > 0$ \cite{BK}:
\begin{theorem}[\cite{BK}] 
    \label{thm:3DM-B}
	There is a large enough integer $B$ such that the following is true for any small constant $\epsilon > 0$. Given a 3DM instance $(X, Y, Z, E)$ where $|X| = |Y| = |Z| = n$ and every $v \in X \cup Y \cup Z$ appears in between $2$ and $B$ triples in $E$, it is NP-hard to distinguish between the following two cases:
    \begin{itemize}
    	\item $(X, Y, Z, E)$ is a $yes$-instance : there is a matching $M \subseteq E$ of size $n$,
    	\item $(X, Y, Z, E)$ is a $no$-instance : every matching $M \subseteq E$ has size at most $\big(\frac{97}{98} + \epsilon\big)n$.
    \end{itemize}
\end{theorem}

\noindent{\bf Maximum $k$-Set Packing}\ \  In the Maximum \( k \)-Set Packing problem, we are given a collection of \( n \) sets, each containing at most \( k \) elements from a given ground set. The objective is to find the largest collection of pairwise disjoint sets. This problem generalizes the \( k \)-Dimensional Matching problem, which involves finding a maximum matching in a \( k \)-partite \( k \)-uniform hypergraph. In particular, Maximum \( 3 \)-Set Packing generalizes 3-Dimensional Matching. The best-known approximation algorithm for the Maximum \( k \)-Set Packing problem achieves a \(\frac{k+1+\epsilon}{3}\)-approximation ratio, given by Cygan \cite{C13}, with subsequent improvements by Fürer and Yu \cite{FY} on the algorithm's running time dependence on \(\epsilon\).   \smallskip


\noindent{\bf Other Notations}\ \ 
Given a real vector $x \in \R^{\mathbb{D}}$ for some domain $\mathbb{D}$, we use $\supp(x):=\{i \in \mathbb{D}: x_i \neq 0\}$ to denote the set of its coordinates with non-zero values. For any subset $U \subseteq \mathbb{D}$, we define $x(U):= \sum_{i \in U} x_i$, unless specified otherwise. 

%% file: overview.tex
\section{Overview of Techniques}
\label{sec:overview}

In this section, we provide a brief overview of the techniques used in our results. \smallskip

\noindent{\bf $2$-Approximation Algorithm for PFCT}\ \ We formulate an LP relaxation by considering a linear objective $\sum_{i \in S, j \in T} \frac{x_{ij}}{b_j} \cdot f_i$. This is upper bounded by its actual cost $\sum_{i \in S, j \in T} \ceil{\frac{x_{ij}}{b_j}} \cdot f_i$ of $x$. The LP can be solved using a greedy algorithm that processes sources in descending orders of $f_i$ values, and sinks in descending order of $b_j$ values. The solution obtained will not contain a ``crossing’’ pair of edges. Using this property, we show that the gap between the actual cost and its LP cost is small. \smallskip

\noindent {\bf $(\frac65+\epsilon)$-Approximation Algorithm for PFCT}\ \  Our goal is to partition $S \cup T$ into the maximum number of balanced sets, where a set is balanced if its total supply equals its total demand. The cost of a solution is $n + m$ minus the number of balanced sets. The maximization problem is a special case of the set-packing problem. Intuitively, the optimum solution in a bad instance of the PFCT-U problem should contain many small balanced sets. Fortunately, for the $k$-set-packing problem, the case where each given set has size at most $k$, a local search algorithm can give a $(k+1+\epsilon)/3$-approximation \cite{FY}. Then we restrict our attention to balanced sets of sizes $3, 4, 5, 6$. For each size $k$, we run the local search $(k+1 + \epsilon)/3$-approximation algorithm and output the best solution. The final approximation ratio $\frac65$ is then obtained by a small linear program.  \smallskip

\noindent{\bf Directed-Steiner-Tree-Hardness of PFCT}\ \ To prove the hardness of PFCT, we introduce an intermediate problem called the PFCT-Digraph problem, where we are given a directed graph and the flows can only be sent along the edges. The reduction from Directed Steiner Tree to PFCT-Digraph is immediate: we just treat the root as a source with supply $k$, and every terminal as a sink with demand $1$, and assign a fixed cost of $1$ to each edge. We then show that PFCT-Digraph is equivalent to the original PFCT problem, by splitting each vertex $v$ into an edge $(v_\rmout, v_\rmin)$. \smallskip

\noindent{\bf Set-Cover-Hardness of FCT-S}\ \ To prove the hardness of FCT-S, we only need to use the special case where each $c_{ij}$ is either $0$ or $\infty$. In this setting, the linear cost ratios determine which edges can be used. The reduction from Set Cover follows the similar vertex-splitting approach as in the previous reduction. Since the fixed costs depend only on the sources, we can only use Set Cover (rather than Directed Steiner Tree) as the basis for the reduction, resulting in a weaker hardness result. \smallskip

\noindent{\bf Reduction from 3DM-B to PFCT-U}\ \ For a given 3DM-B instance $(X, Y, Z, E)$, we randomly assign each element $v \in X \cup Y \cup Z$ a demand $b_v$, from the integer interval $(\Delta, 2\Delta]$. For every triple $ijk \in E$, we create a source $ijk$ with supply $a_{ijk} := b_i + b_j + b_k$. The balanced set $\{i, j, k, ijk\}$ will correspond to the triple $ijk \in E$. A dummy sink is added to ensure total supply equals total demand. Ideally, a solution to the 3DM-B instance should choose the maximum number of balanced sets of the form $\{i, j, k, ijk\}$, and put the remaining sources and sinks into a big balanced set.  To prevent ``cheating'', we choose a sufficiently large $\Delta$ so that all the $b_v$ values for $v \in X \cup Y \cup Z$ are sufficiently ``independent’’:  no two small disjoint multi-sets of $X \cup Y \cup Z$ have the same total $b$ value. This will make cheating expensive. Note that it is important to have a bounded frequency $B$, so that the total number of sources and sinks is not too big compared to the total number of balanced sets in the target solution, which is 1 plus the size of the maximum 3-dimensional matching. \smallskip

\noindent{\bf $2$-Approximation for FCT-S}\ \ The 2-approximation algorithm for FCT-S is straightforward. We first ignore the fixed costs to obtain a solution $x$. As the fixed costs are uniformly 1,  the total fixed cost is simply the number of edges we used.  As $\supp(x)$ is a forest, this yields a 2-approximation for the fixed-cost component, leading to an overall 2-approximation. \smallskip

\noindent{\bf EPAS for PFCT Parameterized by $n$}\ \ For the PFCT problem, a PTAS is easy to obtain when $n = O(1)$: we can guess the set $P$ of $n/\epsilon$ most expensive edges in the solution, and use the $2$-approximation greedy algorithm. The difference between the actual cost and the LP cost of the greedy solution can be bounded by $n$ times the cost of the most expensive undecided edge. This is small compared to the total cost of $P$.  To improve the running time to $g(n, \epsilon)\cdot \poly(m)$, we reduce the number of candidates for the $n/\epsilon$ most expensive edges.  This is achieved by discretizing the fixed costs, and partitioning the sinks into $\left(O\left(\frac{\log m}{\epsilon}\right)\right)^{n}$ classes such that all sinks in a class have the same incident fixed cost vector. Each class can be treated as a PFCT-S instance. In this case, we show there is an approximately optimum solution where the $n/\epsilon$ most expensive edges are incident to a few sinks with the largest demands. This reduces the size of the candidate edges for the set $P$. 

%% file: PFCT-S-2-approx.tex
\section{$2$-Approximation for Pure Fixed Charge Transportation with Sink-Independent Fixed Costs (PFCT-S)}
\label{sec:PFCT-S-2-approx}

In this section, we give our $2$-approximation for the Pure Fixed Charge Transportation with Sink-Independent Costs (PFCT-S) problem. Recall that we are given sources $S$ with $|S| = n$, sinks $T$ with $|T| = m$, supply vector $a \in \Z_{>0}^S$, demand vector $b \in \Z_{>0}^T$ with $a(S) = b(T)$ and fixed cost vector $(f_i \geq 0)_{i \in S}$.  

We rename the sources as $[n]$ and sort them in descending order of $f_i$ values. Also, we rename the sinks as $[m]$ and sort them in descending order of $b_j$ values. So we have \(f_1 \geq f_2 \geq \cdots \geq f_n\) and \(b_1 \geq b_2 \geq \cdots \geq b_m\).

We consider the following linear program for the problem.
\begin{equation}
    \min \sum_{i \in S, j \in T} \frac{x_{ij}}{b_j} \cdot f_i \qquad \text{s.t.} \quad x \in \calX \label{LP}
\end{equation}
Recall that $\calX$ is defined in Section~\ref{subsec:formal-definition}. The actual cost of the solution $x$ is $\sum_{i \in S, j \in T}\ceil{\frac{x_{ij}}{b_j}}\cdot f_i$, which is non-linear in $x$. In the LP, we use the linear function $\sum_{i \in S, j \in T}\frac{x_{ij}}{b_j}\cdot f_i$ as the objective. Clearly, this is upper bounded by the actual cost.  \medskip

We prove the following claim, which will lead to the greedy algorithm for solving the LP. 
\begin{claim}
    \label{claim:no-crossing}
    Focus on a solution $x$ to \eqref{LP}. Assume $x_{ij'} > 0, x_{i'j} > 0$ for some $i < i'$ and $j < j'$.  Let $\epsilon = \min\{x_{ij'}, x_{i'j}\}$. Then the following operation will maintain the validity of $x$ without increasing its cost to \eqref{LP}: decrease $x_{ij'}$ and $x_{i'j}$ by $\epsilon$, and increase $x_{ij}$ and $x_{i'j'}$ by $\epsilon$. 
\end{claim}
\begin{proof}
    Clearly, the operation will maintain the feasibility of $x$. It increases the cost by
    \begin{align*}\displaystyle
        \epsilon\left(\frac{f_i}{b_j} + \frac{f_{i'}}{b_{j'}} - \frac{f_i}{b_{j'}} - \frac{f_{i'}}{b_j}\right) = \epsilon(f_ i - f_{i'})\left(\frac{1}{b_j} - \frac{1}{b_{j'}}\right) \leq 0.
    \end{align*}
    The last inequality holds as $f_i \geq f_{i'}$ and $b_j \geq b_{j'}$. 
\end{proof}

Therefore, there is an optimum solution $x$ to \eqref{LP}, whose support does not contain a ``crossing'' pair of edges as described in Claim~\ref{claim:no-crossing}.  Such a solution is unique, and can be obtained by the following greedy algorithm: 
\begin{algorithm}[H]
    \caption{Greedy Algorithm for PFCT-S}
    \label{alg:PFCT-S-greedy}
    \begin{algorithmic}[1]
        \State $x_{ij} \gets 0$ for all $i \in S, j \in T$
        \State $i \gets 1, j \gets 1$
        \While{$i \leq n$ and $j \leq m$}
            \State $x_{ij} \gets \min\{a_i, b_j\}, a_i \gets a_i - x_{ij}, b_j \gets b_j - x_{ij}$
            \If{$a_i = 0$} $i \gets i + 1$\EndIf
            \If{$b_j = 0$} $j \gets j + 1$\EndIf
        \EndWhile
    \end{algorithmic}
\end{algorithm}

Till the end of this section, we call the $x$ obtained by the greedy algorithm the greedy solution, and we fix this $x$.  We output $x$ as the solution to the PFCT-S instance. 



For real $t \in (0, b([m])]$, let $\pi(t)$ be the smallest $j$ such that $b([j]) \geq t$, i.e., the smallest $j$ such that the first $j$ sinks have total demand at least $t$.  Let $f_{n+1} = 0$ for convenience.

\begin{lemma}
    \label{lemma:PFCT-S-opt}
    The optimum solution of the PFCT instance has cost at least $\sum_{i = 1}^{n} (f_{i} - f_{i+1}) \pi(a([i]))$.
\end{lemma}

\begin{proof}
    We analyze the cost of the optimum solution as follows. When some sink is connected to a source $i$, we split the cost $f_i$ into segments $(f_i - f_{i+1}) + (f_{i+1} - f_{i+2}) + \cdots + (f_n - f_{n+1})$. The segment $f_i - f_{i+1}$ will be paid if we connect a sink to a source in $[i]$. The number of sinks connected to $[i]$ is at least $\pi(a[i])$ by the definition of $\pi$ and that $b_1 \geq b_2 \geq \cdots \geq b_m$. The lemma then follows.
\end{proof}

\begin{lemma}
    \label{lemma:PFCT-S-greedy}
    The cost of $x$ to the PFCT-S instance is at most $\sum_{i = 1}^{n} (f_{i} - f_{i+1}) \pi(a([i])) + \sum_{i = 2}^n f_i$. 
\end{lemma}

\begin{proof}
    Similar to the proof of Lemma~\ref{lemma:PFCT-S-opt}, when $x$ incurs a cost of $f_i$, we split it into segments $(f_i - f_{i+1}) + (f_{i+1} - f_{i+2}) + \cdots + (f_n - f_{n+1})$. So, the segment $f_{i} - f_{i + 1}$ is paid whenever we make a connection to sources in $[i]$.  Notice that the greedy algorithm will only connect $[i]$ to sinks in $[\pi(a[i])]$. As $\supp(x)$ is a forest, the number of edges incident to $[i]$ in $\supp(x)$ is at most $\pi(a[i]) + i - 1$. So, the cost of $x$ to the PFCT-S instance is at most 
    \begin{flalign*}
       && &\quad \sum_{i = 1}^n (f_i - f_{i+1}) \big(\pi(a[i]) + i - 1\big) = \sum_{i = 1}^n (f_i - f_{i+1}) \pi(a[i])  + \sum_{i, i' \in [n]: i' \leq i-1} (f_i - f_{i+1})&&\\
       &&  &= \sum_{i = 1}^n (f_i - f_{i+1}) \pi(a[i])  + \sum_{i' = 1}^n \sum_{i = i'+1}^n(f_i - f_{i+1}) = \sum_{i = 1}^n (f_i - f_{i+1}) \pi(a[i]) + \sum_{i' = 1}^n f_{i'+1}. && \qedhere
    \end{flalign*}
\end{proof}

Combining Lemmas \ref{lemma:PFCT-S-opt} and \ref{lemma:PFCT-S-greedy} and that the cost of the optimum solution is at least $\sum_{i = 2}^n f_i$, we conclude that the cost of $x$ is at most $2$ times optimum cost. This finishes the proof of Theorem~\ref{thm:PFCT-S-2}. \medskip

The following corollary is implied by Lemmas \ref{lemma:PFCT-S-opt} and \ref{lemma:PFCT-S-greedy}. It will be used in our EPAS for the PFCT problem. For convenience, we prove it here. 
\begin{coro}
    \label{coro:compare}
    Suppose we have two PFCT-S instances with the same sources $[n]$, same supplies $a_1, a_2, \cdots, a_{n}$ and same costs $f_1 \geq f_2 \geq \cdots \geq f_{n}$. One instance has $m$ sinks with demands $b_1, b_2, \cdots, b_m$, and the other has $m'$ sinks with demands $b'_1, b'_2, \cdots, b'_{m'}$.  For every $t \in (0, a([n])]$, let $\pi(t)$ be the smallest $j$ such that $b([j]) \geq t$ let $\pi'(t)$ be the smallest $j$ such that $b'([j]) \geq t$. Assume $\pi'(t) \leq \pi(t) + \Delta$ holds for some integer $\Delta \geq 0$ and every $t \in (0, a([n])]$. 
    
    Then, the cost of the greedy solution to the second instance is at most the cost of the optimum solution to the first instance plus $\Delta f_1 + \sum_{i = 2}^n f_i$.
\end{coro}

\begin{proof}
    By Lemma~\ref{lemma:PFCT-S-greedy}, the cost of the greedy solution to the second instance is at most $\sum_{i = 1}^n (f_i - f_{i+1}) \pi'(a[i]) + \sum_{i=2}^n f_i$. By Lemma~\ref{lemma:PFCT-S-opt}, the optimum cost to the first instance is at least $\sum_{i = 1}^n (f_i - f_{i+1}) \pi(a[i])$.  The difference of the two quantities is 
    \begin{align*}
        \sum_{i=1}^n (f_i - f_{i+1}) (\pi'(a[i]) - \pi(a[i])) + \sum_{i = 2}^n f_i \leq \sum_{i = 1}^n (f_i - f_{i+1})\cdot \Delta  + \sum_{i = 2}^n f_i = \Delta f_1  +  \sum_{i = 2}^n f_i.
    \end{align*}
    The inequality used that $\pi'(t) \leq \pi(t) + \Delta$ for every $t$. 
\end{proof}

%% file: PFCT-U-1.2-approx.tex
\section{$(\frac{6}{5}+\epsilon)$-Approximation for Pure Fixed Charge Transportation with Uniform Fixed Costs (PFCT-U)}

\label{sec:PFCT-U-65-approx}

In this section, we consider the most restricted case of the FCT problem, PFCT-U; that is, the Pure Fixed Charge Transportation with Unit Fixed Costs problem. Our main result is a $\left(\frac65 + \epsilon\right)$-approximation for the problem. 

Notice that we can easily make the support of the solution a forest, without increasing the cost.  Since there are no linear costs and all edges have the same fixed cost, the problem is equivalent to partitioning $S \cup T$ into as many \emph{balanced sets} as possible: a subset $V \subseteq S \cup T$ is said to be balanced if $a(V \cap S) = b(V \cap T)$.  The cost of the partition is precisely $m + n$ minus the number of sets in the partition.  We remark that though the two problems are equivalent from the perspective of exact algorithms, the approximation ratio for one problem does not convert to the same ratio for the other. 

Whenever there is some $i \in S$ and $j \in T$ with $a_i = b_j$, we can WLOG let $\{i, j\}$ be a set in the partition. If any solution puts $i$ in some set $V$ and $j$ in a different set $V'$ in the partition, we can replace $V$ and $V'$ with $\{i, j\}$ and $V \cup V' \setminus \{i, j\}$ in the partition; this does not change the cost of the solution. Notice that $V \cup V' \setminus \{i, j\}$ is also balanced. By removing $i$ from $S$ and $j$ from $T$, we obtain a harder instance from the perspective of approximation algorithms.  Therefore, from now on, we can assume in the instance, we have no balanced sets of size $2$.  \medskip

With the connection, we can reduce our problem to the $k$-set packing problem. For every $k \in \{3, 4, 5\}$, we construct a $k$-set packing instance defined over the ground set $S \cup T$ and the family $\calS$ of all balanced sets $V \subseteq S \cup T$ of size at most $k$. We then run the algorithm of \cite{FY} to obtain a $\frac{k+1+\epsilon}{3}$-approximate solution $\calT \subseteq \calS$ for the instance, which naturally gives a solution of cost $m + n - |\calT|$ to the PFCT-U instance. Our final solution is the best one we obtained over all $k \in \{3, 4, 5\}$.  

%
%
\paragraph{Analysis of Approximation Ratio}
We then analyze the approximation ratio of the algorithm. We formulate the following factor revealing LP, with variables $r, z, x_3, x_4, x_5$ and $x_6$:
    \begin{equation}
        \max \qquad r \label{LP:revealing}
    \end{equation}\vspace*{-25pt}
    
\noindent\begin{minipage}{0.5\textwidth}
    \begin{align}  
        z - (x_3 + x_4 + x_5 + x_6) &= 1 \label{LPC:define-z} \\  
        (3x_3 + 4x_4 + 5x_5 + 6x_6) - z &\leq 0  \label{LPC:z-big} \\ 
        r - \Big(z - \frac{3}{4} x_3\Big) &\leq 0 \label{LPC:y3}
    \end{align}
\end{minipage}
\noindent\begin{minipage}{0.5\textwidth}
    \begin{align}
        r - \Big(z - \frac{3}{5} (x_3 + x_4)\Big) &\leq 0 \label{LPC:y4} \\ 
        r - \Big(z - \frac{3}{6} (x_3 + x_4 + x_5)\Big) &\leq 0 \label{LPC:y5} \\  
        x_3, x_4, x_5, x_6, z &\geq 0
    \end{align}
\end{minipage}\medskip

We prove the following lemma:
\begin{lemma}
    Let $r^*$ be the value of LP \eqref{LP:revealing}. The approximation ratio of the algorithm is at most $\big(1 + \frac\epsilon4\big) r^*$.
\end{lemma}
\begin{proof}
    Focus on any instance of the PFCT-U problem $(S, T, (a_i)_{i \in S}, (b_j)_{j \in T})$ with $|S| = n$ and $|T| = m$. We let $\calT^*$ be optimum partition of $S \cup T$ into balanced sets, i.e, the one that maximizes the number of sets.  Let $c_3, c_4, c_5$ and $c_6$ respectively be the number of subsets of size $3$, $4$, $5$ and \emph{at least} $6$ in $\calT^*$.  Recall that we assumed there can not be balanced subsets of size $2$.  Then, the cost of the optimum solution is $\opt := m + n - (c_3  + c_4 + c_5 + c_6)$, and the cost of the solution given by our algorithm is at most
    \begin{align*}
        \min\left\{m + n - \frac{4+\epsilon}{3}\cdot c_3,\quad m + n - \frac{5+\epsilon}{3}\cdot (c_3 + c_4),\quad m + n - \frac{6+\epsilon}{3}\cdot (c_3 + c_4 + c_5)\right\}.
    \end{align*}
    This holds as the optimum solution for the $3$-set (resp.\ $4$-set and $5$-set) packing instance has value at least $c_3$ (resp.\ $c_3 + c_4$ and $c_3 + c_4 + c_5$). 

    We show that the approximation ratio for this instance is at most $\big(1+\frac\epsilon4\big)r^*$, by giving a feasible solution to LP~\eqref{LP:revealing}:
    \begin{align*}
        x_{k'} &= \frac{c_{k'}}{\opt}, \forall k' \in \{3, 4, 5\}, \quad z = \frac{m+n}{\opt}, \\
        \quad \text{and} \quad r &= \min\Big\{z-\frac34x_3, \quad z-\frac35(x_3+x_4), \quad z-\frac36(x_3+x_4+x_5)\Big\}.
    \end{align*}

    Notice that $m + n = \opt + \sum_{k'=3}^6 c_{k'}$. Dividing both sides by $\opt$ gives us $z = 1 + \sum_{k'=3}^6 x_{k'}$, which is \eqref{LPC:define-z}. Also $m+n \geq \sum_{k' = 3}^6 k' c_{k'}$ as $\calT^*$ is a partition of $S \cup T$. Dividing both sides of the inequality by $\opt$, we have $z \geq \sum_{k'=3}^6 k' x_{k'}$, which is \eqref{LPC:z-big}.  By the way we define $r$, \eqref{LPC:y3}, \eqref{LPC:y4} and \eqref{LPC:y5} are satisfied.  The non-negativity constraints hold trivially. 

    The approximation ratio of our algorithm for the instance is 
    \begin{flalign*}
        && &\quad\frac{1}{\opt}\cdot\min \left\{m+n - \frac{4+\epsilon}{3}\cdot c_3, \quad m+n - \frac{5+\epsilon}{3}\cdot (c_3 + c_4), \quad m+n - \frac{6+\epsilon}{3}\cdot (c_3 + c_4 + c_5)\right\}&&\\
        && &=\min\left\{z - \frac{4+\epsilon}{3} x_3,\quad
        z - \frac{5+\epsilon}{3} (x_3 + x_4),\quad
        z - \frac{6+\epsilon}{3} (x_3 + x_4 + x_5)
        \right\}&&\\
        && &\leq \big(1+\frac\epsilon4\big)\cdot \min\left\{z - \frac{4}{3} x_3,
        z - \frac{5}{3} (x_3 + x_4),
        z - \frac{6}{3} (x_3 + x_4 + x_5)
        \right\} \quad = \quad (1+\frac\epsilon4)r \quad\leq\quad (1+\frac\epsilon4)r^*.  &&\qedhere
    \end{flalign*} 
\end{proof}

We show that the value of LP~\eqref{LP:revealing} is $r^* = \frac65$, acheived by the following solution: 
\begin{align*}
    \left(x_2 = 0, x_3 = \frac4{15}, x_4 = \frac1{15}, x_5 = \frac1{15}, x_6 = 0, z = \frac7{5},r = \frac65\right).
\end{align*}
First, it is easy to check that the solution satisfies all the constraints in the LP. 

We show that this is optimal, by considering the dual LP. The optimal dual solution is $\alpha = \frac65, \beta = \frac15, y_3 = \frac4{15}, y_4 = \frac13, y_5 = \frac25$, where the 5 variables correspond to constraints \eqref{LPC:define-z}, \eqref{LPC:z-big}, \eqref{LPC:y3}, \eqref{LPC:y4} and \eqref{LPC:y5} respectively.  Formally, for any solution $(x_3, x_4, x_5, x_6, z, r)$ satisfying the constraints, we have 
\begin{align*}
    r &\leq \frac{4}{15} \Big(z - \frac34 x_3\Big) + \frac1{3}\Big(z - \frac35(x_3 + x_4)\Big) + \frac25\Big(z-\frac36(x_3 + x_4 + x_5)\Big)\\
    &= z - \big(\frac15 + \frac1{5} + \frac15\big) x_3 - \big(\frac1{5} + \frac15\big)x_4 - \frac15 x_5 = z - \frac35 x_3 - \frac25x_4 - \frac15 x_5\\
    &= \frac65 (z - (x_3 +x_4 + x_5 + x_6)) - \frac15(z - 3x_3 - 4x_4 - 5x_5  - 6x_6) \leq \frac65.
\end{align*}
Therefore, the value of the LP is exactly $\frac65$. This proves that our algorithm is an $(\frac65 + \epsilon)$-approximation. 

We remark that the approximation ratio remains $\frac65 + \epsilon$ even if we consider bigger values of $k$. The primal solution $(x_2 = 0, x_3 = \frac4{15}, x_4 = \frac1{15}, x_5 = \frac1{15}, x_k = 0,\forall k\geq 6, z = \frac7{5},r = \frac65)$ remains a valid solution of value $\frac65$.

%% file: hardness.tex
\section{Hardness of Pure Fixed Charge Transportation (PFCT) via Reduction from Directed Steiner Tree (DST)}

\label{sec:PFCT-hardness}

In this section, we prove Theorem~\ref{thm:PFCT-DST-hard} and Corollary~\ref{coro:PFCT-DST-hard} by giving the reduciton from Directed Steiner Tree (DST) to Pure Fixed Charge Transportation (PFCT).  In Section~\ref{subsec:PFCT-Digraph}, we introduce the Pure Fixed Charge Transportation on Directed Graphs (PFCT-Digraph) problem, which is seemingly more general than the PFCT problem. But in Section~\ref{subsec:digraph-equivalence}, we show that the two problems are indeed equivalent. In Section~\ref{subsec:PFCT-from-DST}, we give a reduction from DST problem to PFCT-Digraph, and thus PFCT. 

\begin{figure}
    \centering
    \includegraphics[width=0.8\linewidth]{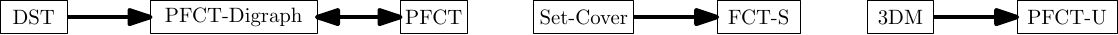}
    \caption{Reductions used in our hardness results.}
    \label{fig:reductions}
\end{figure}

\subsection{The Pure Fixed Charge Transportation on Digraphs (PFCT-Digraph) Problem}
\label{subsec:PFCT-Digraph}

We introduced the Pure Fixed Charge Transportation in a Directed Graph (PFCT-DAG) problem, which has been studied previously \cite{RW, OW}.  We are given a directed graph (DAG) $G = (V, E)$, a set $S \subseteq V$ of sources, and a set $T \subseteq V$ of sinks such that $S \cap T = \emptyset$. 

As in PFCT, each source $i$ has $a_i \in \Z_{>0}$ units of supply, each sink $j$ has $b_j \in \Z_{>0}$ units of demand, and $\sum_{i \in S}a_i = \sum_{j \in T}b_j$. Every edge $e \in E$ has a fixed cost $f_e \in \R_{\geq 0}$.

The goal of the problem is to send flows from $S$ to $T$ across $G$, so that each $i \in S$ sends exactly $a_i$ units of flow, and each $j \in T$ receives exactly $b_j$ units of flow.  If an edge $e \in E$ is carrying a non-zero units of flow, then we pay a cost of $f_e$. 
The goal of the problem is to minimize the total cost we pay.  Formally, we need to output a flow vector $x \in \R_{\geq 0}^E$ satisfying that $\sum_{e \in \delta^{\rmout}_v} x_e - \sum_{e \in \delta^{\rmin}_v} x_e$ is $a_v$ if $v \in S$, $-b_v$ if $v \in T$ and $0$ otherwise. Our goal is to minimize $\sum_{e \in E} \big(\mathbf{1}_{x_e > 0} \cdot f_e\big)$. 

Clearly,  PFCT is a special case of PFCT-Digraph where the graph is a directed bipartite graph from the sources $S$ to the sinks $T$. 

\subsection{Equivalence of PFCT and PFCT-Digraph} 
\label{subsec:digraph-equivalence}
In this section, we show a reduction from PFCT-Digraph to PFCT, establishing the equivalence of the two problems. We assume we are given a PFCT-Digraph instance $(G = (V, E), S, T, a \in \Z_{>0}^S, b \in \Z_{>0}^T, f \in \R_{\geq 0}^E)$.  WLOG we assume $S$ does not have incoming edges and $T$ does not have outgoing edges. For every vertex $v \in V \setminus (S \cup T)$, we apply the following splitting operation on $v$. See Figure~\ref{fig:splitting} for an illustration of the operation. 

\begin{figure}[h]
    \centering
    \includegraphics[width=0.4\linewidth]{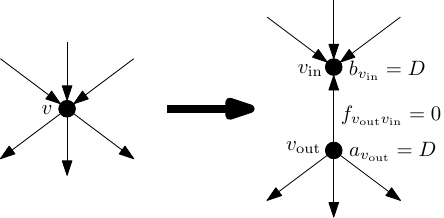}
    \caption{Splitting operation in reduction from PFCT-Digraph to PFCT.}
    \label{fig:splitting}
\end{figure}

We break $v$ into $v_{\rmin}$ and $v_{\rmout}$ and add an edge from $v_{\rmout}$ to $v_{\rmin}$ of cost 0. All the incoming edges of $v$ now go to $v_{\rmin}$, and all the outgoing edges of $v$ are from $v_{\rmout}$; we do not change the $f$ values of these edges. $v_{\rmin}$ now has a demand $D$ and $v_{\rmout}$ has a supply $D$, for a sufficient large integer $D$. ($D = \sum_{i \in S}a_i = \sum_{j \in T} b_j$ suffices.) After this conversion, the graph becomes a directed bipartite graph, where edges are from $S \cup \{v_{\rmout}: v \in V \setminus (S \cup T)\}$ to $T \cup \{v_{\rmin}: v \in V \setminus (S \cup T)\}$.

It is easy to see the equivalence between the two instances. Focus on a solution $x \in \R_{\geq 0}^E$ to the PFCT-Digraph instance. Then, in the solution for the PFCT instance, we keep the $x$ values of these edges unchanged. The $x$ value the edge $(v^{\rmout}, v^{\rmin})$ for $v \in V \setminus (S \cup T)$ is defined as $D - \sum_{e \in \delta^{\rmout}_v} x_e = D - \sum_{e \in \delta^{\rmin}_v} x_e$. This is non-negative if $D$ is large enough. The edges $(v^{\rmout}, v^{\rmin})$ have $f$ value being $0$, and thus the cost of the solution to the PFCT instance is the same as the cost of $x$ for the original PFCT-Digraph instance.  Similarly, we can convert a solution for the PFCT instance to one for the PFCT-Digraph instance with the same cost.

We also remark that the reduction naturally extends to the case where we have linear costs on the edges; we only need to set the $c$ value of $v_{\rmout\rmin}$ to $0$. This establishes the earlier claim that the single-commodity uncapacitated fixed-charge network problem is indeed equivalent to the FCT problem.

\subsection{Reduction from DST to PFCT-Digraph} 
\label{subsec:PFCT-from-DST}

In this section, we reduce the Directed Steiner Tree (DST) problem to the PFCT-Digraph problem with two sources. We are given a DST instance $(G = (V, E), c, r, T, k = |T|)$, and we assume each terminal $t \in T$ has only one incoming edge and no outgoing edges. We create a PFCT-Digraph instance as follows. Start from $G = (V, E)$ and let the costs $f_e$ in the PFCT-Digraph problem equal to the costs $c_e$ in the DST-O problem. We let the root $r$ be a source and the terminals $T$ be the sinks; the source $r$ has $a_r := k$ units of supply and each sink $t \in T$ has $b_t:=1$ unit of demand. 

To see the equivalence between the PFCT-Digraph instance and the DST instance, it suffices to notice that the following property: There is an optimal solution $x \in \R_{\geq 0}^E$ for the PFCT-Digraph instance such that the underlying undirected graph for $\supp(x)$ contains no cycles. If this is not the case, we can take any cycle in the undirected graph, choose a direction for the cycle, and increase or decrease $x$ values of the edges on the cycle depending on whether the edges have the same or opposite direction to the chosen direction, until the $x$ value of some edge on the cycle becomes $0$.  This operation does not increase the cost of $x$ to the PFCT-Digraph instance. So, $\supp(x)$ is precisely a directed Steiner tree with the root $r$ and the terminals $T$. 

We then finish the proof of Theorem~\ref{thm:PFCT-DST-hard}. Given a DST instance, we can construct an equivalent PFCT instance that preserves the optimum cost. If the DST instance has $n$ vertices, then PFCT instance has at most $n$ sources and at most $n$ sinks. 
Therefore, an $\alpha(\max\{n, m\})$-approximation for the PFCT instance leads to an $\alpha(n)$-approximation for the DST instance. Conversely, the $\Omega(\log^{2-\epsilon}(n))$-hardness for DST of \cite{HK03} implies an $\Omega(\log^{2-\epsilon}(\max\{n, m\}))$-hardness of approximation for PFCT, which proves Corollary~\ref{coro:PFCT-DST-hard}.

\section{APX-Hardness of Pure Fixed Charge Transportation with Uniform Fixed Costs (PFCT-U) from 3DM-B}
\label{sec:PFCT-U-from-3DM-B}
In this section, we prove the APX-hardness of PFCT-U (Theorem~\ref{thm:PFCT-U-APX-hard}) by a reduction from the 3DM-B problem. We let $\alpha = \frac{97}{98}$ and $B$ be the constants in Theorem~\ref{thm:3DM-B}, and let $\epsilon > 0$ from the theorem be small enough.  Let $B' = \ceil{3(3 + B - 4(\alpha + \epsilon))/\epsilon}$. Suppose we are given a 3DM instance $(X, Y, Z, E)$ as stated in the theorem. Let $n = |X| = |Y| = |Z|$, and $m = |E|$. We assume $n$ is large enough. As every element appears in at most $B$ triples, we have $m \leq Bn$.

Let $\Delta = 2(6n+1)^{B'}$. For each element $v \in X \cup Y \cup Z$, we let $b_v$ be an integer chosen from $(\Delta, 2\Delta]$ uniformly at random. 
\begin{lemma}
    \label{lemma:sampling-good}
    With probability at least $1/2$, the following event happens. For every $h \in \Z^{X \cup Y \cup Z}$ with $|h|_1 \in [1, B']$, we have $\sum_{v \in X \cup Y \cup Z}h_v b_v \neq 0$.
\end{lemma}
\begin{proof}
    A vector $h \in Z^{X \cup Y \cup Z}$ with $|h|_1 \leq B'$ can be recorded as a sequence of $B'$ operations starting from the $0$ vector, where each operation is increasing $h_v$ by 1 for some $v \in X \cup Y \cup Z$, decreasing $h_v$ by 1 for some $v \in X \cup Y \cup Z$, or doing nothing. So, there are at most $(2\times 3n + 1)^{B'} = (6n + 1)^{B'}$ different vectors $h \in Z^{X \cup Y \cup Z}$ with $|h|_1 \leq B'$.
    
    For a fixed vector $h \in Z^{X \cup Y \cup Z}$ with $|h|_1 \in [1, B']$, the probability that $\sum_{v \in X \cup Y \cup Z}h_v b_v = 0$ happens is at most $\frac 1\Delta$. Then applying union bound gives the lemma. 
\end{proof}

We choose $(b_v)_{v \in X \cup Y \cup Z}$ values satisfying the condition of the lemma.  Then we construct a PFCT-U instance $(S, T, a, b)$ as follows. First, every $v \in X \cup Y \cup Z$ is a sink, with the chosen $b_v$ value.  Every triple $ijk \in E$ is a source in the instance, with $a_{ijk} := b_i + b_j + b_k$.  Finally, we add a ``dummy'' sink $t$ with $b_t := \sum_{ijk \in E}a_{ijk} - \sum_{v \in X \cup Y \cup Z}b_v$ to make the instance balanced.  So, in the PFCT-U instance, we have $S = E$ and $T = X \cup Y \cup Z \cup \{t\}$.

Till the end of this section, we say a set $U \subseteq S \cup T$ is balanced if $a(S \cap U) = b(T \cap U)$. 
\begin{claim}
    If the 3DM-B instance is a yes-instance, then the optimum cost of the PFCT-U instance is at most $2n + m$. 
\end{claim}
\begin{proof}
    Let $M \subseteq E$ be a perfect matching for the 3DM-B instance. We can partition $S \cup T$ into $n + 1$ balanced sets naturally: For every $ijk \in M$, we select the balanced set $\{i, j, k, ijk\}$. Then we put the remaining elements in the final balanced set, which is $\{t\} \cup E \setminus M$. The cost of the solution to the PFCT-U instance is $3n + m + 1 - (n+1) = 2n + m$.
\end{proof}

Then we focus on the case where the 3DM-B instance is a no-instance. We say a balanced set $U \subseteq S \cup T$ is canonical if it is $\{i, j, k, ijk\}$ for some $ijk \in E$. 
\begin{lemma}
    \label{lemma:canonical-or-large}
    Let $U \subseteq S \cup T$ be a non-empty balanced set. Then $U$ is either the union of some disjoint canonical sets, or has size at least $B'/3$.
\end{lemma}
\begin{proof}
    First, assume $t \notin U$.  As $U$ is balanced, we have $\sum_{ijk \in U \cap S} (b_i + b_j + b_k) - \sum_{v \in U \cap T}b_v  = 0$. We view each $b_v$ in the equation as variables and focus on the coefficients $h_v$'s of $b_v$'s. If $h \in \Z^{X \cup Y \cup Z}$ is the $0$-vector, then $U$ must be the union of some disjoint canonical sets.  If $h \neq 0$, then by Lemma~\ref{lemma:sampling-good}, we have $|h|_1 \geq B' + 1$ non-zero entries, implying that the size of $U$ is at least $B'/3$. 

    Consider the case $t \in U$. $b_t = \sum_{ijk \in E} (b_i + b_j + b_k) -\sum_{v \in X \cup Y \cup Z}b_v \geq \sum_{V \in X \cup Y \cup Z} b_v \geq 3n\Delta$.  Each $ijk \in E$ has $a_{ijk} \leq 6\Delta$. Thus $|U| \geq n/2 + 1$. As we assumed that $n$ is big enough, we have $|U| \geq B'/3$.
\end{proof}

\begin{lemma}
    If the 3DM-B instance is a no-instance, then the cost of the PFCT-U instance is at least $(3 - \alpha - 2\epsilon)n + m$. 
\end{lemma}
\begin{proof}
    Focus on the optimum partition $\calP$ of $S \cup T$ into balanced sets.  If a balanced set $U \in \calP$ is the union of at least $2$ canonical sets, then $\calP$ is not optimum since breaking $U$ into canonical sets will make $\calP$ better. Therefore, every $U \in \calP$ is either a canonical set, or has size at least $B'/3$, by Lemma~\ref{lemma:canonical-or-large}.  As the maximum matching for the 3DM instance has size at most $(\alpha + \epsilon)n$, the number of canonical sets in $\calP$ is at most $(\alpha + \epsilon)n$. Therefore, we have 
    \begin{align*}
        |\calP| \leq (\alpha + \epsilon) n + \frac{3n + m + 1 - 4(\alpha + \epsilon) n}{B'/3}.
    \end{align*}
    The cost of the solution is at least
    \begin{align*}
        3n + m + 1 - |\calP| \geq (3 - \alpha - \epsilon)n + m - \frac{3(3n + m - 4(\alpha + \epsilon)n)}{B'}
    \end{align*}

    As $m \leq Bn$ and $B' \geq \frac{3 (3 + B - 4(\alpha + \epsilon))}{\epsilon}$, we have $\frac{3(3n + m - 4(\alpha+\epsilon) n)}{B'} \leq \epsilon n$. So, the cost is at least $(3 - \alpha - 2\epsilon)n + m$.
\end{proof}

   So, the multiplicative gap between the costs for the no-instance case and the yes-instance case is at least 
    \begin{align*}
        \frac{(3-\alpha-2\epsilon)n + m}{2n + m}.
    \end{align*}
    If $2\epsilon < 1 - \alpha$, then $3 - \alpha - 2\epsilon > 2$. The gap is at least $\frac{(3-\alpha-2\epsilon)n + Bn}{2n + Bn} = \frac{3-\alpha-2\epsilon + B}{2 + B}$, which is an absolute constant larger than $1$. This finishes the proof of the APX-hardness for the PFCT-U problem in Theorem~\ref{thm:PFCT-U-APX-hard}.

%% file: EPAS.tex
\section{Efficient Parameterized Approximation Scheme (EPAS) for Pure Fixed Charge Transportation (PFCT) Problem}
\label{sec:EPAS}
In this section we give an Efficient Parameterized Approximation Scheme (EPAS) for the PFCT problem, parameterized by the number $n$ of sources. For the case $n = 2$ and $f_{ij} = 1$, the problem is equivalent to the subset sum problem: if the subset sum instance is feasible, the optimum cost of the PFCT instance is $m$; otherwise it is $m + 1$. This rules out FPTAS for the case, unless P $=$ NP.  As a starting point, we show a PTAS for the case where $n$ is a constant in Section~\ref{subsec:PTAS-for-PFCT}. This will serve as our building block for the EPAS. 

\subsection{$(1+\epsilon)$-Approximation in $(nm)^{O(n/\epsilon)}$ Time}
\label{subsec:PTAS-for-PFCT}
We first give a PTAS for the PFCT problem with running time $(nm)^{O(n/\epsilon)}$. Throughout, we assume $1/\epsilon$ is an integer. We guess the set $P$ of $2n/\epsilon$ most expensive edges in the optimum solution; there are at most $(nm)^{\frac{2n}{\epsilon}}$ possibilities for this set. We output the best solution over all guesses. So, from now on, 
%
we assume we correctly guessed the set $P$. Let $F:=\min_{ij\in P} f_{ij}$ be the cost of the cheapest edge in $P$. We then solve the following linear program to obtain $x$, and return $x$ as the solution for the PFCT instance. 

\begin{equation}
    \min \quad \sum_{ij\in P} f_{ij}+\sum_{ij\in (S\times T) \setminus P}\frac{x_{ij}}{b_j} \cdot f_{ij} \quad \text{s.t.} \quad 
\end{equation}
\begin{align*}
    x \in \calX, \quad
    x_{ij} = 0, \forall ij \in (S \times T) \setminus P \text{ with } f_{ij} > F
\end{align*}


The value of the LP is at most the cost of the optimum solution to the PFCT instance. Again, we can assume that $\supp(x)$ is a forest. For a leaf sink $j$ in the forest, and its unique incident edge $ij$, we must have $x_{ij} = b_j$. Therefore, there are at most $2n$ edges $ij \notin P$ with $x_{ij} \in (0, b_j)$. The maximum  $f_{ij}$ value among all these edges is at most $F$. So, the actual cost of $x$ to the PFCT instance is at most the LP value plus $2n F$. As $2nF$ is at most $\epsilon$ times the cost of $P$, which is at most $\epsilon$ times the optimum cost, we have that $x$ is a $(1+\epsilon)$-approximate solution. This finishes the description of the PTAS for the problem when $n = O(1)$.  \medskip

To improve the running time to $g(n, \epsilon) \cdot \poly(m)$, our goal is to reduce the number of candidates for the $n/\epsilon$ most expensive edges.  In Section~\ref{subsec:reduce-candidate-set-PFCT-S}, we show that there exists a $(1+\epsilon)$-approximate solution, where the $n/\epsilon$ most expensive edges are incident to the $n/\epsilon$ sinks with the largest demands, for the PFCT-S instance. This will greatly reduce the candidate set of edges and result in the desired running time.  For the PFCT problem, we partition the sinks into many classes, where each class of sinks has the same cost vector incident to the $n$ sources. Using a standard discretization step, we can bound the number of classes by $\left(O\left(\frac{\log m}{\epsilon}\right)\right)^n$. Then, the instance restricted to each class can be viewed as a PFCT-S instance. We bound the size of the candidate set for each class, leading to an overall bound for all classes. This algorithm is described in Section~\ref{subsec:EPAS-PFCT}.

\subsection{The Case of Sink-Independent Fixed Costs}
\label{subsec:reduce-candidate-set-PFCT-S}

In this section we show that we can reduce the size of the candidate set of the most expensive edges, with a loss of $1+\epsilon$ in the approximation ratio. We main theorem we prove is:
\begin{theorem}
    \label{thm:reduce-candidate-set}
    Suppose we are given a PFCT-S instance defined by the set $S = [n]$ of sources, the set $T = [m]$ of sinks, the supply vector $(a_i)_{i \in [n]}$, the demand vector $(b_j)_{j \in [m]}$ and the fixed cost vector $(f_i)_{i \in [n]}$.  Assume $f_1 \geq f_2 \geq \cdots \geq f_n$ and $b_1 \geq b_2 \geq \cdots \geq b_m$.  Let $x^*$ be the optimum solution to the instance. Assume $\supp(x^*)$ is the union of disjoint stars centered at sources; that is, every sink $j \in [m]$ is incident to exactly one source in $\supp(x^*)$.  
    
    Let $\epsilon > 0$ be a small enough constant. Then there is a $(1+\epsilon)$-approximate optimum solution $x$, such that the $n/\epsilon$ most expensive edges in $\supp(x)$ are incident to sinks in $[n/\epsilon]$. 
\end{theorem}

    \begin{figure}
        \centering
        \includegraphics[width=0.3\linewidth, page=1]{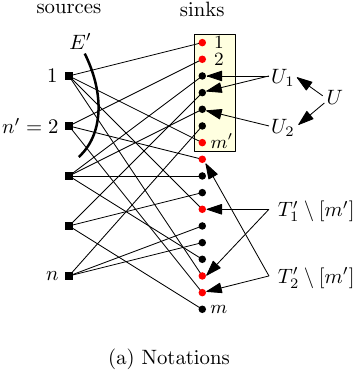} 
        \hfill
        \includegraphics[width=0.3\linewidth, page=2]{modify-x.pdf} 
        \hfill
        \includegraphics[width=0.3\linewidth, page=3]{modify-x.pdf} 
        \caption{Notations used in proof of     Theorem~\ref{thm:reduce-candidate-set}. In the three sub-figures, the $n$ sources are denoted by squares, the $m$ sinks are denoted by dots, the sinks $[m']$ are those in the yellow rectangle, and sinks $T'$ are represented by red dots.       The lines connecting the sources and the sinks in sub-figure (a) denote $\supp(x^*)$. The lines in sub-figures (b) and (c) denote the pre-selected edges for the first and second PFCT-S instances respectively. The blue thick lines show the difference between the two instances.}
        \label{fig:modify-x}
    \end{figure}
\begin{proof}
    For every $i \in [n]$, let $T_i$ be the set of sinks connected to $i$ in $\supp(x^*)$. So, $T_1, T_2, \cdots, T_n$ form a partition of $[m]$ as $\supp(x^*)$ is a forest of stars centered at sources. Let $m' = n/\epsilon$; $[m']$ is the set of the $m'$ sinks with the largest $b$ values. 

    Let $E' \subseteq \supp(x^*)$ be the $m'$ most expensive edges in $\supp(x^*)$, and $T'$ be the set of sinks incident to $E'$. Let $T'_i = T' \cap T_i$ for every $i \in [n]$. As $\supp(x^*)$ is a forest of stars and $f_1 \geq f_2 \geq \cdots \geq f_n$, there exists some $n' \in [n]$ such that $T'_i = T_i$ for every $i \in [n' - 1]$, $\emptyset \neq T'_i \subseteq T_i$ and $T'_i = \emptyset$ for every $i > n'$. 
    

    
    Notice that $|T'| = m'$, and thus $|T' \setminus [m']| = |[m'] \setminus T'|$. Moreover the $b$-value of any sink in $[m'] \setminus T'$ is at least the $b$-value of any sink in $T' \setminus [m']$. We can use a greedy algorithm to construct disjoint subsets $U_1, U_2, \cdots, U_{n'}$ of $[m'] \setminus T'$ such that the following happens for every $i \in [n']$.
    \begin{enumerate}[leftmargin=*, label=(P\arabic*)]
        \item $b(U_i) \geq b(T'_i\setminus [m'])$.
        \item $b(U_i \setminus j) < b(T'_i \setminus [m'])$ for every $j \in U_i$. This implies $|U_i| \leq |T'_i \setminus [m']|$.
    \end{enumerate}
    
Let $U = U_1 \cup U_2 \cup \cdots \cup U_{n'} \subseteq [m'] \setminus T'$; notice that it is possible that $U \subsetneq [m'] \setminus T'$.   See Figure~\ref{fig:modify-x}(a) for an illustration of definitions we made so far.\medskip

We define two PFCT-S instances, on which we apply Corollary~\ref{coro:compare}.  Both instances are residual instances of the given instance.  The first residual instance is obtained by assigning demands in $T'$, but not $[m] \setminus T'$, according to $x^*$. See Figure~\ref{fig:modify-x}(b) for the instance.  Therefore, in the first PFCT-S instance,
\begin{itemize}
    \item every sink $i > n'$ has $a_i$ units of supply, and 
    \item every $j \in [m] \setminus T'$ has $b_j$ units demand. 
\end{itemize}\medskip

The second residual instance is obtained from the first one by switching the flows sent to $T'_i \setminus [m']$ to $U_i$, for each $i \in [n']$. By (P1) and (P2), we can guarantee that at most one sink in $U_i$ has a positive demand remaining. See Figure~\ref{fig:modify-x}(c) for an illustration of the instance. Therefore, in the second PFCT-S instance,
\begin{itemize}
    \item every sink $i > n'$ has $a_i$ units of supply,
    \item every $j \in (m', m] \cup ([m'] \setminus T' \setminus U)$ has $b_j$ units demand, and
    \item for every $i \in [n']$, there exists a $j \in U_i$ with $b(U_i) - b(T'_i \setminus [m'])$ units demand, if the amount is positive.
\end{itemize} \medskip


We then compare the two PFCT-S instances using Corollary~\ref{coro:compare}. We define $\pi(t)$ and $\pi'(t)$ as in the corollary.  Let $D = |[m']\setminus T'| - |U| = |T' \setminus [m']| - |U| = \sum_{i \in [n']} (|T'_i \setminus [m'
]| - |U_i|)$. 

\begin{lemma}
    $\pi'(t) \leq \pi(t) + D + n'$ for every $t \geq (0, b(U)]$.
\end{lemma}
\begin{proof}
    It suffices for us to show that, for any $V^1 \subseteq [m] \setminus T'$, there exists a set $V^2 \subseteq [m]$ of size at most $|V^1| + D + n'$ such that the total demand in $V^1$ in the first PFCT-S instance is at most the total demand in $V^2$ in the second PFCT-S instance.  Notice that the sinks in $\big((m', m] \setminus T'\big) \cup \big([m'] \cap T'\big) \cup \big([m'] \setminus T' \setminus U\big)$ have the same demands in both instances. It suffices for us to focus on sets $V^1 \subseteq U$ and require $V^2 \subseteq U \cup (T' \setminus [m'])$.
    
    Consider any subset $V^1 \subseteq [m'] \setminus T'$; define $V^1_i = U_i \cap V^1$ for every $i \in [n']$. We shall construct a set $V^2$ of sinks as follows. For every $i \in [n']$, we add to $V^2$ the sink in $U_i$ with a positive supply in the second instance, if it exists. Then, we add to $V^2$ an arbitrary subset of $T'_i \setminus [m']$ of size $|T'_i \setminus [m']| - |U_i \setminus V^1_i|$. Clearly, the total demand of these sinks in the second instance is at least $b(V^1_i)$. The total demand of $W^2 := U_i \cup (T'_i \setminus [m'])$ in the second instance is equal to $b(U_i)$. The $b$-value of any sink in $T'_i \setminus [m']$ is at most the $b$-value of any sink in $U_i$. So, after we remove $|U_i \setminus V^1_i|$ sinks in $(T'_i \setminus [m'])$ from $W^2$, and $|U_i \setminus V^1_i|$ from $U_i$, the demand in the former set is at least the demand in the latter, w.r.t their respective PFCT-S instances. 
    
    Therefore, the total demand at $V^2$ in the second instance is at least $b(V^1)$. Moreover, $|V^2| \leq \sum_{i \in [n']}(|T'_i \setminus [m']| - |U_i \setminus V^1_i| + 1) = \sum_{i \in [n']}(|T'_i \setminus [m']| - |U_i|) + \sum_{i \in [n']} |V^1_i| + n' = D + |V^1| + n'$. The lemma then follows.
\end{proof}

Therefore, using the Corollary~\ref{coro:compare}, the greedy solution we construct for the second instance has cost at most that of the optimum solution for the first instance, plus $(D + n')f_{n'+1} + f_{n'+2} + f_{n'+3} + \cdots + f_n \leq (D + n - 1)f_{n'+1}$.  The cost for the pre-selected edges in the first instance, is that in the second instance, plus 
$$\displaystyle
    \sum_{i \in [n']} f_i (|U_i| - |T'_i \setminus [m']|) \geq f_{n'+1} \sum_{i \in [n']}(|U_i| - |T'_i \setminus [m']|) = D f_{n'+1}.$$

We let $x$ be the solution for the original instance, using the greedy solution for the second residual instance. That is, $x$ contains the pre-selected edges for the second instance, and the edges given by the greedy algorithm. $x^*$ contains the pre-selected edges for the first residual instance, the optimum solution for the instance. Therefore, the cost of $x$ minus the cost of $x^*$ is at most $(D + n - 1)f_{n'+1} - Df_{n'+1} = (n-1)f_{n'+1}$ more than the cost of $x^*$.  This is at most $\epsilon$ times the cost of $E'$, which is at most $\epsilon$ times cost of $x^*$. 

Finally, in the solution $x$, there is no ``crossing'' between sinks $[m']$ and $(m', m]$: there is no pair $ij, (i', j') \in \supp(x)$ such that $i > i', j \in [m']$ and $j' \in [m] \setminus [m']$.  This holds as pre-selected edges of the instance go from $[n']$ to $[m']$. The greedy algorithm only makes connections to sources $(n', n]$ and it does not create crossings. 
%
\end{proof}

Now, we shall remove the assumption in Theorem~\ref{thm:reduce-candidate-set} that $x^*$ is the union of stars centered at sources. This is summarized in the following corollary.
\begin{coro}
    \label{coro:reduce-candidate-set}
    Let $S, n, T, m, (a_i)_{i \in [n]}, (b_j)_{j \in [m]}$ and $(f_i)_{i \in [n]}$ be as in Theorem~\ref{thm:reduce-candidate-set}; they define a PFCT-S instance. Let $\epsilon > 0$ be a small enough constant. Then there is an $(1+\epsilon)$-approximate solution $x$, such that the $n/\epsilon - n$ most expensive edges in $\supp(x)$ are incident to $[n/\epsilon + n]$.
\end{coro}

\begin{proof}
    Let $x^*$ be the optimum solution to the PFCT-S instance. We modify the instance and $x^*$ as follows. Assume some sink $j$ is incident to $o \geq 2$ sources $i_1, i_2, \cdots, i_o$ in $\supp(x^*)$, and the flow sent by the $o$ sources are respectively $y_1 := x^*_{i_1j}, y_2 := x^*_{i_2j}, \cdots, y_o = x^*_{i_oj}$, with $y_1 + y_2 + \cdots + y_o = b_j$.  We can then modify the PFCT-S instance by splitting $j$ into $o$ sinks, with demands being $y_1, y_2, \cdots, y_o$ respectively. We modify the solution $x^*$ naturally: The $q$-th newly created sink receives $y_q$ units flow from $i_q$.  Notice that this operation can only make the instance harder: any solution for the new instance can be converted into one for the old instance with no greater cost. On the other hand, the cost of $x^*$ does not change.  We repeat the operation until every sink is incident to exactly one source in $\supp(x^*)$. Notice that the $x^*$ is a forest at the beginning, and thus the splitting operations will increase the number of sinks by at most $n - 1$. 
    
    We then apply Theorem~\ref{thm:reduce-candidate-set} on the new instance; so there is a $(1+\epsilon)$ approximate solution $x$, such that the $n/\epsilon$ most expensive edges $E'$ in $\supp(x)$ are incident to sinks in $[n/\epsilon]$. By merging sinks, $x$ can be transformed into a solution for the original instance with no greater cost. Some edges may be merged due to the merging of sinks. But the merging operation does not change the fact that $E'$ remains the most expensive edges after merging. $|E'|$ may be decreased by at most $n - 1$.  Due to the merging of sinks, the parent sinks of the $n/\epsilon$ sinks with the largest $b$ values in the new instance is a subset of the $n/\epsilon + n$ sinks with the largest $b$ values in the original instance. 
\end{proof}

\subsection{Efficient Parameterized Approximation Scheme (EPAS) for PFCT}
\label{subsec:EPAS-PFCT}

In this section, we use Corollary~\ref{coro:reduce-candidate-set} built in Section~\ref{subsec:reduce-candidate-set-PFCT-S} to design an Efficient Parameterized Approximation Scheme (EPAS) for the Pure Fixed Charge Transportation (PFCT) problem. Notice that now every pair $ij \in S \times T$ has a $f_{ij}$ value. We can use standard techniques to discretize the $f_{ij}$ values so that there are only $O(\log_{1+\epsilon}m)$ different $f_{ij}$ values. By guessing we assume we know the cost $F$ of the most expensive edge used in the optimum solution. If an edge has $f_{ij} > F$ we can change $f_{ij}$ to $\infty$. If an edge has $f_{ij} < \epsilon F/(n + m)$, we can change it to $0$. This will only incur a loss of $1+\epsilon$ in the approximation ratio. For the remaining edges, we round $f_{ij}$ up to its nearest integer power of $1+\epsilon$. Therefore, there are only $O(\log_{1+\epsilon}m) = O\left(\frac{\log m}{\epsilon}\right)$ different $f_{ij}$ values, after this processing. We only lose a factor of $(1+\epsilon)^2$ in the approximation ratio. 

We partition $T$ into classes, where two sinks $j$ and $j'$ are in the same class if $f_{ij} = f_{ij'}$ for every $i \in S$. So, there are at most $\left(O\left(\frac{\log m}{\epsilon}\right)\right)^n$ different classes of sinks. Let $M = \left(O\left(\frac{\log m}{\epsilon}\right)\right)^n$ be this number. 

\begin{lemma}
    We can efficiently construct a set $T' \subseteq T$ of at most $M(n/\epsilon + n)$ sinks such that the following holds. There is a $(1+\epsilon)$-approximate solution to the PFCT instance, where the $n/\epsilon - n$ most expensive edges in $\supp(x)$ are incident to $T'$.
\end{lemma}
\begin{proof}
    Let $T^1, T^2, \cdots, T^M \subseteq T$ be the $M$ classes.  Focus on the optimum solution $x^*$ for the instance. We can partition $x^*$ into $M$ pieces, with the $p$-th piece being $x^*|_{S \times T^p}$. Then, we can view $x^*|_{S \times T^p}$ as a solution to the PFCT-S instance with sources $S$, sinks $T^p$, demand vector $(b_j)_{j \in T^p}$. \begin{itemize}
        \item The supply for a source $i \in S$ is $\sum_{j \in T^p} x^*_{ij}$,
        \item The cost $f_i$ for a source $i \in S$ is $f_{ij}$ for any $j \in T^p$. As all $f_{ij}$'s are the same for all $j \in T^p$, $f_i$ is uniquely defined.
    \end{itemize}
    By Corollary~\ref{coro:reduce-candidate-set}, there is a $1+\epsilon$-approximate solution $x^p$ to the PFCT-S instance, where the $n/\epsilon - n$ most expensive edges are incident to the $n/\epsilon + n$ sinks in $T'^p$ with the largest demands; call these sinks $T^p$. Putting $x^p$'s for all $p \in [M]$ together gives a solution $x$ for the whole PFCT instance. The $n/\epsilon - n$ most expensive edges in $\supp(x)$ are incident to $T':=\union_{p \in [M]}T'^p$. 
\end{proof}

Therefore, we can guess the $n/\epsilon - n$ most expensive edges in $x$. There are at most $(Mn(n/\epsilon + n))^{n/\epsilon} = \left(O\left(\frac {n^2}\epsilon\right)\right)^{n/\epsilon}\cdot \left(O\left(\frac{\log m}{\epsilon}\right)\right)^{n^2/\epsilon}$ different possibilities for the set. Once we obtain the set, we can use the greedy algorithm to obtain a solution with approximation ratio $1 + \frac{n-1}{n/\epsilon - n} \leq 1 + 2\epsilon$.  Overall, the approximation ratio we obtain is $1 + O(\epsilon)$.

The running time of the algorithm is $\left(O\left(\frac {n^2}\epsilon\right)\right)^{n/\epsilon}\cdot \left(O\left(\frac{\log m}{\epsilon}\right)\right)^{n^2/\epsilon} \cdot \poly(m)$.  We focus on the term $\left(O\left(\frac{\log m}{\epsilon}\right)\right)^{n^2/\epsilon}$.  When $n \leq \left(\epsilon \log m\right)^{1/3}$, this is at most $m^{O(1)}$. When $n > \left(\epsilon\log m\right)^{1/3}$, the bound is at most $\left(O\left(\frac{n^3}{\epsilon^2}\right)\right)^{n^2/\epsilon} = \left(\frac n\epsilon\right)^{O(n^2/\epsilon)}$.  Therefore, the running time of the algorithm is at most $\left(\frac{n}{\epsilon}\right)^{O(n^2/\epsilon)}\cdot \poly(m)$.

%% file: discussion.tex
\section{Discussion}

In this paper, we initiated a systematic study of variants of the Fixed Charge Transportation problem, and provided a complete characterization of the existence of $O(1)$-approximation algorithms for them.  Although our techniques are simple, we believe this work opens a new research direction on the problem.

There are several interesting open problems. For the PFCT-S and FCT-U problems, can we achieve an approximation ratio better than 2? For the FCT-S problem, is it possible to obtain a polylogaritmic approximation? For the PFCT problem,  can we show that it  is not significantly harder than the Directed Steiner Tree problem? For example, by allowing quasi-polynomial time algorithms, can we obtain a polylogarithmic approximation?  Finally, can we establish much stronger hardness of approximation results for the two problems, or the most general FCT problem?

%% file: appendix.tex
\section{$2$-Approximation for Fixed Charge Transportation with Uniform Fixed Costs (FCT-U)}

The Fixed Charge Transportation with Uniform Fixed Costs (FCT-U) problem has a simple $2$-approximation, and thus we include it in the preliminaries. We ignore the fixed costs, and consider the problem of minimizing the linear cost $\sum_{i \in S, j \in T}c_{ij}x_{ij}$.  This is simply a linear program and thus we can obtain an optimum solution efficiently.  WLOG, we can assume the support of the solution $x$ is a forest. If there is a cycle in the support, we can rotate the flow on the cycle in one direction until the flow on some edge becomes $0$.  There are two directions in which we can rotate the cycle, and we can choose the direction so that the operation does not increase the cost.  

As the solution $x$ we obtain is a forest, its fixed cost is at most $m + n - 1$. The fixed cost of any solution is at least $\max\{m, n\} \geq (m + n - 1)/2$. As $x$ minimizes the linear cost, it is a $2$-approximate solution to the FCT-U instance.  This proves Theorem~\ref{thm:FCT-U-2}.

\section{Set-Cover Hardness of Fixed Charge Transportation with Sink-Independent Costs (FCT-S)}
\label{subsec:FCT-S-set-cover}
In this section, we give the reduction from Set Cover to Fixed Charge Transportation with Sink-Independent Costs (FCT-S) problem; this proves Theorem~\ref{thm:FCT-S-Set-Cover-hard}. Indeed, our reduction is to the special case of FCT-S with $c_{ij} \in \{0, \infty\}$ for every $ij \in S \times T$. The mere role of the $c$-vector is to define the set of edges that can be used. 

\paragraph{Set Cover} In the Set Cover problem, we are given the ground set $[n]$, and $m$ sets $S_1, S_2, \cdots, S_m \subseteq [n]$. The goal of the problem is to choose the smallest number of subsets to cover the whole ground set $[n]$. That is, to find the smallest $I \subseteq [m]$ with $\union_{i \in I} S_i = [n]$.  It is well-known that a simple greedy algorithm yields a $(\ln n + 1)$-approximation for the problem \cite{J74, L75}. A long line of research \cite{LY94, F98, DS14} culminated in a tight lower bound by Dinur and Steurer \cite{DS14}, who proved  that it is NP-hard to approximate the Set Cover problem within a factor of $(1 - \epsilon) \ln n$, for any constant $\epsilon > 0$. 


Throughout, we shall view the Set Cover problem as following dominating set problem on a bipartite graph. We are given a bipartite graph $(V, U, E)$ with $|V| = m$ and $|U| = n$, and the goal is to choose the smallest subset $V' \subseteq V$ to \emph{dominate} $U$: $V'$ dominates $U$ if every $u \in U$ is adjacent to at least one vertex in $V'$. By viewing each $v \in V$ as a set, each $u \in U$ as an element, and $vu \in E$ indicates $u$ is in the set $v$, we can see that the problem is indeed the Set Cover problem. \medskip
 
At a high level, we reduce the Set Cover instance to a PFCT-Digraph instance with a single source $s^*$. We then apply the splitting operation in Section~\ref{subsec:PFCT-from-DST} to change this to a FCT-S instance, where we use $c_{ij}$ values to control which edges can be used. The outgoing edges of $s^*$ have fixed costs $1$; the other edges have fixed costs $0$. Thus the costs are sink-independent. See Figure~\ref{fig:set-cover-FCT-S} for an illustration. 
\begin{figure}[h]
    \centering
    \includegraphics[width=0.95\linewidth]{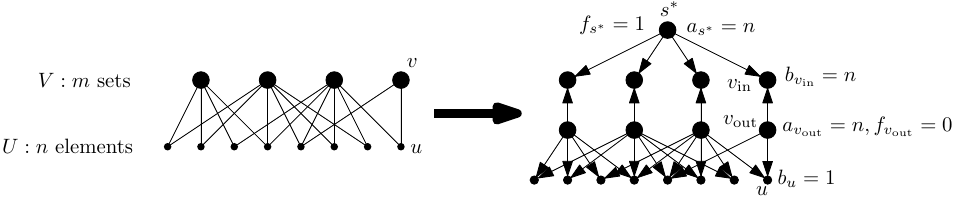}
    \caption{Reduction from Set Cover to FCT-S.}
    \label{fig:set-cover-FCT-S}
\end{figure}

Formally, we introduce a source $s^*$, and a source $v_\rmout$ for every set $v \in V$. So $S = \{s^*\} \cup \{v_\rmout: v \in V\}$. We introduce a sink $v_\rmin$ for every $v \in V$; each $u \in U$ is also a sink. So, $T = \{v_\rmin: v \in V\} \cup U$. As we mentioned, we use the vector $c \in \{0, \infty\}^{S \times T}$ to define the edges that can be used. There is an edge from $v_\rmout$ to $v_\rmin$ for every $v \in V$, an edge from $s^*$ to $v_\rmin$ for every $v \in V$, and an edge from $v_\rmout$ to $u$ for every $(v, u) \in E$, that is, for every $v, u$ such that the set $v$ contains the element $u$.

Then we define the $a, b$ and $f$ values. The source $s^*$ has $a_{s^*} = (m - 1)n$ and $f_{s^*} = 1$. A source $v_\rmout$ has $a_{v_\rmout} = n$ and $f_{v_\rmout} = 0$. A sink $v_\rmin$ has $b_{v_\rmin} = n$, and a sink $u \in U$ has $b_u = 1$. So, the total supply of all sources is equal to the total demand of all sinks. Only $s^*$ has a non-zero $f$-value, which is 1; the other sources have $f$-values being $0$. 

We now show a one-to-one correspondence between a solution $V'$ and a solution for the FCT-S instance. Assume we are given a set $V' \subseteq V$ that dominates $U$ in the bipartite graph $(V, U, E)$; we define a solution to the FCT-S instance with cost $|V'|$. For every $v \in V \setminus V'$, we sent $n$ units flow from $v_\rmout$ to $v_\rmin$; this will satisfy the supplies and demands on the two vertices. Every $u \in U$ is dominated by some $v \in V'$, and we sent one unit flow from $v_\rmout$ to $u$ for this $v$. Notice that every $v_\rmout$ has $n$ units supply, which is sufficient for the demands from $U$. For every $v \in V'$, the remaining unused supply will be sent to $v_\rmin$. After this, the unsatisfied demand at $v_\rmin$ will be satisfied by $s^*$. The cost of this solution is at most $|V'|$, as we only send positive flows from $s^*$ to $v_\rmin$ for $v \in V'$. 

Now focus on a feasible solution to the FCT-S instance. Let $V'$ be the set of vertices $v \in V$ with positive flow sent from $s^*$ to $v_\rmin$; so the cost of the solution is $|V'|$. Then for every $v \in V \setminus V'$, there is no flow sent from $s^*$ to $v_\rmin$, and thus $n$ units flow sent from $v_\rmout$ to $v_\rmin$. This implies that there is no flow sent from $v_\rmout$ to $U$. Therefore, all the demands at $U$ are satisfied by $\{v_\rmout: v \in V'\}$. Therefore, $V'$ dominates $U$.

We then finish the proof of Theorem~\ref{thm:FCT-S-Set-Cover-hard}. To avoid confusion, we use $N$ and $M$ to denote the numbers of sources and sinks in the FCT-S instance we construct, and $n$ and $m$ to denote the numbers of elements and sets in the given Set Cover instance. The hard Set-Cover instances of \cite{DS14} have $m = \poly(n)$. Given the Set Cover instance with $n$ elements and $m = \poly(n)$ sets, we constructed FCT-S instance with $N := n + 1$ sources and $M := m + n$ sinks. If $c$ is small enough, then $c \ln (\max\{N, M\}) = c\ln(m+n) \leq 0.9\ln n$. Then a $c \ln (\max\{N, M\})$-approximation for FCT-S implies a $0.9\ln n$-approximation for Set Cover, contradicting the $(1-\epsilon)\ln n$-hardness result of \cite{DS14}. 

\section{Bi-Criteria Approximation for Fixed Charge Transportation (FCT)}
\label{sec:bicriteria}
    In this section, we prove Theorem~\ref{thm:FCT-bicriteria} by giving the bicriteria-approximation for the FCT problem. Define $p_{ij} = \min\{a_i, b_j\}$ for every $i \in S, j \in T$. We solve the LP of minimizing $\sum_{i \in S, j \in T}\big(c_{ij} + \frac{f_{ij}}{p_{ij}}\big)x_{ij}$ subject to $x \in \calX$. Then for every $i \in S, j \in T$, we define $y_{ij} := \frac{x_{ij}}{p_{ij}} \in [0, 1]$. 

    The support of $y_{ij}$ is a forest. WLOG, we assume it is a tree $\bfT$; if this is not the case, we can consider each tree separately. We shall get a new solution $y'$ as follows. We root $\bfT$ at an arbitrary vertex in $S \cup T$. Focus on every $v$ and all its child edges $E'$. We define $(a|b)_v$ to be $a_v$ if $v \in S$, or $b_v$ if $v \in T$. The vector $y'|_{E'}$ satisfies the following properties:
    \begin{itemize}
        \item $y'_e = y_e$ if $y_e \geq \epsilon$,
        \item $y'_e \in \{0, \epsilon\}$ if $y_e < \epsilon$,
        \item $\sum_{e \in E'} p_e y'_e - \sum_{e \in E'} p_e y_e \in (-\epsilon \cdot (a|b)_v, 0]$, and
        \item $\sum_{e \in E'} (c_e p_e + f_e) y'_e \leq \sum_{e \in E'} (c_e p_e + f_e) y_e$.
    \end{itemize} 
    That is, if $y_e \in [\epsilon, 1]$, we then keep it unchanged. Otherwise, $y_e \in [0, \epsilon)$ and we need to change it to either $0$ or $\epsilon$. We can do this without changing $\sum_{e \in E'}p_e y_e$ or increasing $\sum_{e \in E'}(c_e p_e + f_e) y_e$, until when there is at most one edge $e\in E'$ with $y_e \in (0, \epsilon)$. Then we simply change $y_e$ to $0$. This will not increase $\sum_{e \in E'} (c_e p_e + f_e) y_e$, and $\sum_{e \in E'} p_e y'_e$ will be decreased by at most $\epsilon p_e \leq \epsilon (a|b)_v$. So the four properties can be satisfied.

    Once we obtain the vector $y'$, we can define $x'_{ij} = p_{ij} y'_{ij}$ for every $i \in S, j \in T$.  We output $x'$ as our solution. For every $i \in S$, we have $\sum_{j \in T} p_{ij} y'_{ij} - \sum_{j \in T} p_{ij} y_{ij}  \leq \epsilon a_i$, where the difference can only come from the parent edge of $i$. On the other hand,
    \begin{align*}
         \sum_{j \in T} p_{ij} y_{ij} - \sum_{j \in T} p_{ij} y'_{ij} \leq \epsilon\cdot a_i + \epsilon \cdot a_i = 2\epsilon \cdot a_i. 
    \end{align*}
    The first $\epsilon \cdot a_i$ comes from the parent edge of $i$, and the second one comes from the child edges.  Notice that $\sum_{j \in T}p_{ij}y_{ij} = \sum_{j \in T} x_{ij} = a_i$. Therefore, for every $i \in S$, we have
    \begin{align*}
        \sum_{j \in T}x'_{ij} = \sum_{j \in T}  p_{ij} y'_{ij} \in ((1-2\epsilon) a_i,(1 + \epsilon) a_i].
    \end{align*}
    Similarly, for every $j \in T$, we have 
    \begin{align*}
        \sum_{i \in S}x'_{ij} \in [(1-2\epsilon) b_j, (1 +\epsilon) b_j].
    \end{align*}
    
    Therefore, we have a flow $x' \in \R_{\geq 0}^{S \times T}$ where the amount of flow sent by each $i\in S$ is between $(1-2\epsilon)a_i$ and $(1+\epsilon)a_i$, and the flow received by each $j \in T$ is between $(1-2\epsilon)b_j$ and $(1+\epsilon)b_j$. We can scale the all $x'$ values within a factor between $\frac{1}{1+\epsilon}$ and $\frac{1}{1-2\epsilon}$ so that 
    \begin{itemize}
        \item every source $i$ sends exactly $a_i$ units flow
        \item every sink $j$ receives between $\frac{1-2\epsilon}{1+\epsilon}\cdot b_j$ and $\frac{1+\epsilon}{1-2\epsilon}\cdot b_j$ units of flow. 
    \end{itemize}
    We can scale down the $\epsilon$ value at the beginning so that in the end so that the flow received by each $j$ is within $(1\pm \epsilon)b_j$.  The cost of $x'$ to the FCT instance is at most $O(1/\epsilon)$ times its cost to the LP, which is at most $O(1/\epsilon)$ times the cost of optimum solution to the LP, which is at most $O(1/\epsilon)$ times the cost of the optimum solution to the FCT instance. This finishes the proof of Theorem~\ref{thm:FCT-bicriteria}. 